\documentclass[english,notitlepage,runningheads]{llncs}

\usepackage{babel}
\usepackage{mathtools}
\usepackage{amsmath}
\usepackage{amsfonts}
\usepackage{amssymb}
\usepackage{hyperref}
\usepackage{xspace}
\usepackage{booktabs}
\usepackage{semantic}
\usepackage[usenames,dvipsnames]{xcolor}
\usepackage{listings}
\usepackage{qcircuit}
\usepackage{wrapfig}
\usepackage{cleveref}

\newcommand{\F}{\mathbb{F}_2}

\newcommand{\N}{\mathbb{N}}

\newcommand{\ket}[1]{|#1\rangle}
\newcommand{\bra}[1]{\langle#1|}

\newcommand{\cnot}{\mathrm{CNOT}}

\newcommand{\x}{\mathbf{x}}

\title{Sized Types for low-level Quantum Metaprogramming}
\author{Matthew Amy\inst{1}\orcidID{0000-0003-3514-420X}}
\institute{University of Waterloo, Waterloo, Canada \\ \email{meamy@uwaterloo.ca}}

\begin{document}

\maketitle

\abstract{One of the most fundamental aspects of quantum circuit design is the concept of \emph{families} of circuits parametrized by an instance size. As in classical programming, metaprogramming allows the programmer to write entire families of circuits simultaneously, an ability which is of particular importance in the context of quantum computing as algorithms frequently use arithmetic over non-standard word lengths. In this work, we introduce metaQASM, a typed extension of the openQASM language supporting the metaprogramming of circuit families. Our language and type system, built around a lightweight implementation of \emph{sized types}, supports subtyping over register sizes and is moreover type-safe. In particular, we prove that our system is strongly normalizing, and as such any well-typed metaQASM program can be statically unrolled into a finite circuit.

\keywords{Quantum programming, Circuit description languages, \\ Metaprogramming.}
}

\section{Introduction}

Quantum computers have the potential to solve a number of important problems, including integer factorization \cite{s94}, quantum simulation \cite{l96}, approximating the Jones polynomial \cite{ajl06} and unstructured searching \cite{g96} asymptotically faster than the best known classical algorithms. These algorithms are typically described abstractly and make heavy use of classical arithmetic such as modular exponentiation. To make such algorithms concrete, efficient, reversible implementations of large swaths of a classical arithmetic and computation is needed -- moreover, due to the limited space constraints and special-purpose nature of quantum circuits, these operations are typically needed in a multitude of bit sizes.

In part due to the increasing viability of quantum computing and the scaling of NISQ \cite{p18} devices, there has been a recent explosion in quantum programming tools. Such tools range from software development kits (e.g., Qiskit \cite{qiskit}, ProjectQ \cite{sht18}, Strawberry Fields \cite{kiqbaw18}, Pyquil \cite{scz16}) to Embedded domain-specific languages (e.g., Quipper \cite{glrsv13}, Qwire \cite{prz17}, $Q\ket{SI}$ \cite{lwglhdy17}) and standalone languages and compilers (e.g., QCL \cite{o00}, QML \cite{ag05}, ScaffCC \cite{jpkhlcm15}, Q\# \cite{sgtaghkmpr18}). Going beyond strict programming tools, software for the synthesis, optimization, and simulation of quantum circuits and programs (e.g., Revkit \cite{revkit}, TOpt \cite{hc18}, Feynman \cite{feynman}, PyZX \cite{kw19}, Quantum++ \cite{qpp}, QX \cite{qx}) are becoming more and more abundant.

The proliferation of both hardware and software tools for quantum computing has in turn spurred a need for standardization and portability \cite{hsst18,mr19}. One such standard which has recently grown in popularity is the Quantum Assembly Language and its many various dialects (e.g., openQASM \cite{cbsg17}, QASM-HL \cite{jpkhlcm15}, cQASM \cite{kgahab18}). As a lightweight, modular language for specifying simple quantum circuits, programs with a well-defined syntax, QASM support -- in particular, for the openQASM dialect -- has been built-in to an increasingly large number of software tools, particularly standalone programs like circuit optimizers, as a way to support interoperability. 

One feature that is noticeably lacking in these dialects is the ability to define \emph{families} of quantum circuits parametrized over different register sizes, and by extension to \emph{generate} concrete instances. This creates a barrier for the use of QASM in writing portable libraries of quantum circuit families, particularly for classical operations such as arithmetic. As a result, software designers typically end up re-implementing code -- typically implemented in the host language for EDSLs, and hence not easily re-usable -- for generating instances of simple operations such as adders and multipliers. Alternatively, programmers resort to using other compilers such as Quipper, Q\# or ReVerC \cite{ars17} to generate individual instances, which complicates the compilation or simulation process. While recent progress towards the development of portable libraries of circuit families with high-level non-embedded languages, standardization remains an on-going process, and moreover a low-level approach is preferable in many situations, including as compilation targets and middle-ends.

In this paper we make progress towards the design of a low-level language for quantum programming that supports the metaprogramming of sized circuit families. In particular, we develop a typed extension of the untyped open quantum assembly language (openQASM) with metaprogramming over lightweight \emph{sized types} \`{a} la dependent ML \cite{x01}. Our language, metaQASM, is further shown to be type-safe and strongly-normalizing, while the non-meta fragment is both more expressive than openQASM and admits a simpler syntax, owing to the type system. For the purposes of this paper, we focus on the type system design and metatheory of such a language, leaving implementation to future work.

\subsection{Quantum metaprogramming}

Most QRAM-based quantum programming languages are metaprogramming languages -- called \emph{circuit description languages} -- in that they typically operate by building quantum circuits to be sent in a single batch to a quantum processor. Such quantum circuits can typically be composed, reversed, and depend on the result of classical computations.

In this paper, we are interested in a particular type of quantum circuit metaprogramming, wherein circuit families are parametrized over \emph{shapes} \cite{glrsv13,prz17}, such as the number of input qubits. Existing languages offer varying support for such metaprogramming, either implicitly (e.g., uniform or \emph{transversal} families of circuits in openQASM, iteration and qubit arrays in Q\#), or more explicitly (e.g., the generic \texttt{QData} type-class in Quipper, which can be instantiated via explicit type applications). Our approach differs from previous attempts by explicitly parametrizing registers and circuit families with \emph{size} parameters. We adopt a typed approach for a number of reasons:
\begin{itemize}
	\item it allows the light-weight verification of libraries of circuit generators,
	\item it provides a means of self-documentation, and
	\item it allows explicit generation of sized-specialized instances.
\end{itemize}
The ability to generate instances of circuit families in various sizes \emph{without executing them} is particularly important for the purposes of resource estimation, and for benchmarking tools that operate on fixed-size but arbitrary input circuits, such as circuit optimizers \cite{hsst18}.

As an illustration, given an in-place family of adders written in the style of (imperative) Quipper with the type

\medskip
\centerline{
	\texttt{inplace\_add :: [Qubit] -> [Qubit] -> Circ ()},
}
\medskip

one may wish to generate a static, optimized instance of \texttt{inplace\_add} operating on $2$-qubit registers, using an external circuit optimizer. Doing so requires the specialization to (and serialization of) a function

\medskip
\centerline{
	\texttt{inplace\_add2 :: (Qubit, Qubit) -> (Qubit, Qubit) -> Circ ()}.
}
\medskip

One possible method of generating such a function is to write the body of \texttt{inplace\_add2} using a call to the generic \texttt{inplace\_add} applied to the $4$ input qubits. However, this quickly gets unwieldy, both in the boilerplate code defining a particular instance, and in the large number of parameters.

A more common solution is to use \emph{dummy parameters}, whereby the generic  function is ``applied'' to lists of qubits, which are then taken by the serialization method as meaning arbitrary inputs. For instance, the following Quipper\footnote{The function \texttt{inplace\_add2} could instead be directly generated by writing the adder as \texttt{inplace\_add :: QData qa => qa -> qa -> Circ ()}, then specializing \texttt{qa} to the finite type \texttt{(Qubit, Qubit)} using \emph{type applications}. However, the non-generic serialization functions in Quipper appear to work only for small finite tuple types.} code \cite{quipper} prints out a PDF representation of \texttt{inplace\_add2} using dummy parameters \texttt{qubit :: Qubit}

\medskip
\centerline{
	\texttt{print\_generic PDF inplace\_add [qubit, qubit] [qubit, qubit]}.
}
\medskip

The use of dummy parameters is partly a question of style, though it can cause problems when combining optimizations with \emph{initialized} dummy parameters. In either case, the use explicitly sized circuit families carries further benefits to both readability and correctness \cite{prz17}.

\subsection{Organization}

The remainder of this paper is organized as follows. \Cref{sec:overview} gives a brief overview of quantum computing. \Cref{sec:qasm} reviews the openQASM language and defines a formal semantics for it. \Cref{sec:tqasm,sec:mqasm} extend openQASM with types and metaprogramming capabilities, and finally \Cref{sec:conclusion} concludes the paper.

\section{Quantum computing}\label{sec:overview}

We give a brief overview of the basics of quantum computing. For a more in-depth introduction of quantum computation we direct the reader to \cite{nc00}, while an overview of quantum programming can be found in \cite{g06}.

In the circuit model, the state of an $n$-qubit quantum system is described as a unit vector in a dimension $2^n$ complex vector space. The $2^n$ elementary basis vectors form the \emph{computational} basis, and are denoted by $\ket{\x}$ for bit strings $\x\in\{0,1\}^n$ -- these are called the \emph{classical} states. A general quantum state may then be written as a \emph{superposition} of classical states
\[
  \ket{\psi} = \sum_{\x\in\F^n} \alpha_{\x}\ket{\x},
\]
for complex $\alpha_{\x}$ and having unit norm. The states of two $n$ and $m$ qubit quantum systems $\ket{\psi}$ and $\ket{\psi}$ may be combined into an $n+m$ qubit state by taking their tensor product $\ket{\psi}\otimes\ket{\psi}$. If to the contrary the state of two qubits cannot be written as a tensor product the two qubits are said to be \emph{entangled}.

Quantum circuits, in analogy to classical circuits, carry qubits from left to right along \emph{wires} through \emph{gates} which transform the state. In the unitary circuit model gates are required to implement unitary operators on the state space -- that is, quantum gates are modelled by complex-valued matrices $U$ satisfying $UU^\dagger = U^\dagger U = I$, where $U^\dagger$ is the complex conjugate of $U$. As a result, unitary quantum computations must be \emph{reversible}, and in particular the quantum circuits performing classical computations are precisely the set of reversible circuits.

The standard universal quantum gate set, known as Clifford+$T$, consists of the two-qubit controlled-NOT gate ($\cnot$), and the single-qubit Hadamard ($H$) and $T$ gates. As quantum circuits implement linear operators, we may define the above three gates by their effect on classical states:
\[
	\cnot\ket{x}\ket{y} = \ket{x}\ket{x\oplus y},
	\qquad T\ket{x} = e^{\frac{2\pi i}{8}x}\ket{x},
\]
\[
	H\ket{x}=\frac{1}{\sqrt{2}}\sum_{x'\in\{0,1\}}(-1)^{x\cdot x'}\ket{x'}.
\]
Figure~\ref{fig:toffoli} gives a pictorial representation of a quantum circuit over $\cnot$, $H$, and $T$ gates. $\cnot$ gates are written as a solid dot on their first argument and an exclusive-OR symbol ($\oplus$) on their second argument.

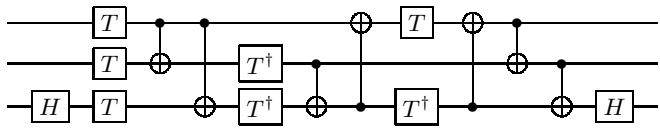
\begin{figure}[b]
\centerline{
\Qcircuit @C=1em @R=.3em {
& \qw & \gate{T} & \ctrl{1} & \ctrl{2} & \qw & \qw & \targ & \gate{T} & \targ & \ctrl{1} & \qw & \qw & \qw \\
& \qw & \gate{T} & \targ & \qw & \gate{T^\dagger} & \ctrl{1} & \qw & \qw & \qw & \targ & \ctrl{1} & \qw & \qw \\
& \gate{H} & \gate{T} & \qw & \targ & \gate{T^\dagger} & \targ & \ctrl{-2} & \gate{T^\dagger} & \ctrl{-2} & \qw & \targ & \gate{H} & \qw 
}
}
\caption{An example of a quantum circuit implementing the Toffoli gate.}
\label{fig:toffoli}
\end{figure}

More general quantum operations include qubit initialization and measurement, which effectively convert between classical and quantum data. As neither operation is unitary and hence not (directly) reversible, we regard them as functions of the classical computer rather than gates in a quantum circuit.

\section{openQASM}\label{sec:qasm}

The open quantum assembly language (openQASM \cite{cbsg17}) is a low-level, untyped imperative quantum programming language, developed as a dialect of the informal QASM language. One of the key additions of the openQASM language is that of \emph{modularity}, in the form of a simple module and import system. As this work is largely concerned with the question of \emph{making this modularity more powerful} -- specifically, to support the modular definition of entire circuit families -- we first give a brief overview of the openQASM language.

The official specification of openQASM can be found in \cite{cbsg17}. Programs in openQASM are structured as sequences of declarations and commands. Programmers can declare statically-sized classical or quantum registers, define unitary circuits (called \emph{gates} in openQASM), apply gates or circuits, measure or initialize qubits and condition commands on the value of classical bits. Gate arguments are restricted to individual qubits, where the application of gates to one or more register \emph{of the same size} is syntactic sugar for the application of a single gate in parallel across the registers. The listing below gives an example of an openQASM program performing quantum teleportation:
\begin{lstlisting}
OPENQASM 2.0;
qreg q[3];
creg c0[1];
creg c1[1];

h q[1];
cx q[1],q[2];
cx q[0],q[1];
h q[0];
measure q[0] -> c0[0];
measure q[1] -> c1[0];
if(c0==1) z q[2];
if(c1==1) x q[2];
\end{lstlisting}

We give a slightly different syntax from the above, and from the concrete syntax \cite{cbsg17}, as it will be more convenient and readable for our purposes. As is common in imperative languages, we leave some of the concrete syntactic classes of openQASM \cite{cbsg17} separate in our formalization -- since all operations in openQASM nominally have unit type, this allows terms with unitary and non-unitary \emph{effects} to be distinguished, without relying on an effect system or monadic types. In particular, terms of the class $U$ of unitary statements represent computations with purely unitary effects, while commands $C$ may have non-unitary effects, such as measurement. Statements of the form

\medskip
\centerline{
	\texttt{$E$($E_1,\dots,E_n$)}
} 
\medskip

represent the application of a unitary gate or named circuit $E$ to the (quantum) arguments $E_1$ through $E_n$. While the openQASM specification includes built-in \texttt{cx} (controlled-NOT) and parametrized single qubit gates \texttt{U}, we drop the parametrized \texttt{U} gate in favour of built-in Hadamard and $T/T^\dagger$ gates \texttt{h} and \texttt{t}/\texttt{tdg}, respectively. 

The commands \texttt{creg}, \texttt{qreg} and \texttt{gate} declare classical registers, quantum registers, and unitary circuits, respectively. The \texttt{if} statement differs from the formal openQASM definition by testing the value of a \emph{single} classical bit, rather than a classical register -- this was done to simplify the semantics of the language. Locations $l$ and values $V$ do not appear directly in openQASM programs, but are used to define the semantics. In particular, values of the form $(l_0, \dots, l_{I-1})$ denote registers and $\lambda x_1, \dots, x_n.U$ denote unitary circuits. We leave out a number of features of openQASM which are orthogonal to the extensions we describe here, namely classical arithmetic and the \texttt{barrier} and \texttt{opaque} terms. We also write parentheses around arguments and parameters.

\begin{figure}[t]
	\begin{tabular}{rrl}
		Identifier $x$ & & \\
		Index $I$ & $::=$ & $i\in\N$ \\
		Expression $E$ & $::=$ & $x$ $\mid$ $x[I]$ \\
		Unitary Stmt $U$ & $::=$ & \texttt{cx}$(E_1,E_2)$ $\mid$ \texttt{h}$(E)$ $\mid$ \texttt{t}$(E)$ $\mid$ \texttt{tdg}$(E)$ $\mid$ $E(E_1,\dots, E_n)$ $\mid$ \texttt{$U_1$; $U_2$} \\
		Command $C$ & $::=$ & \texttt{creg} $x[I]$ $ \mid$ \texttt{qreg} $x[I]$ $\mid$ \texttt{gate $x(x_1,\dots, x_n)$ \{ $U$ \}} \\
		& $\mid$ & \texttt{measure $E_1$ -> $E_2$} $\mid$ \texttt{reset} $E$ $\mid$ \texttt{$U$} \\
		& $\mid$ & \texttt{if($E$==$I$) \{ $U$ \}} $\mid$ \texttt{$C_1$; $C_2$} \\
		& & \\
		Location $l$ & $\in\N$ & \\
		Value $V$ & $::=$ & $(l_0, \dots, l_{I-1})$ $\mid$ $\lambda x_1, \dots, x_n.U$
	\end{tabular}
	\caption{openQASM (abstract) syntax}\label{fig:syntax}
\end{figure}

As no formal semantics of openQASM is given in \cite{cbsg17}, we define an operational semantics in \Cref{fig:semantics}. Our semantics is defined with respect to a \emph{configuration} $\langle S, \sigma, \eta, \ket{\psi}\rangle$, which stores a term $S$ taken from some syntactic class (e.g., $C$, $U$, $E$), an environment $\sigma$ which maps variables to values, a classical heap $\eta$ storing the value of the classical bits, and a quantum state $\ket{\psi}$. Gates applied to qubit $l$ of a quantum state are written by added a subscript to the intended gate, e.g.,
\[
	H_l\ket{\psi} = (I^{\otimes l-1}\otimes H \otimes I^{\otimes n - l})\ket{\psi}
\]
$\sigma[x\leftarrow v]$ denotes the environment mapping $x$ to $v$ or $\sigma(x)$ otherwise, and $S\{X/x\}$ denotes the substitution of $X$ for $x$ in $S$. We assume for convenience that no valid program will run out of classical memory or quantum bits. We say $\langle S, \sigma, \eta, \ket{\psi}\rangle\Downarrow v$ if $S$ reduces to $v$, where the form of $v$ depends on the syntactic class of $S$ -- for instance, expressions evaluate to locations, arrays or circuits while commands produce a new environment, heap and quantum state. Note that we use a call-by-name evaluation strategy, as openQASM has only globally scoped variables.

Rather than give a full probabilistic reduction system to account for measurement probabilities, it suffices for our purposes to make the semantics non-deterministic. In particular, rules are given for both of the possible measurement outcomes in \texttt{measure $E_1$ -> $E_2$}, setting the classical bit to the result $c\in\{0,1\}$ and non-destructively applying the projector $P^c=\ket{c}\bra{c}$ (appropriately normalized) to the measured qubit.

\begin{figure}

Expressions:
	\begin{gather*}
		\inference{x\in\textsf{dom}(\sigma)}{\langle x, \sigma, \eta, \ket{\psi}\rangle\Downarrow \sigma(x)} \qquad
		\inference{\langle x, \sigma, \eta, \ket{\psi}\rangle\Downarrow (l_0, \dots, l_{I'}) \qquad I\leq I'}{\langle x[I], \sigma, \eta, \ket{\psi}\rangle\Downarrow l_I}
	\end{gather*}
Unitary statements:
	\begin{gather*}
		\inference{\langle E, \sigma, \eta, \ket{\psi}\rangle\Downarrow l}
			{\langle \texttt{h}(E), \sigma,\eta, \ket{\psi}\rangle \Downarrow H_{l}\ket{\psi}}	
		\inference{\langle E, \sigma, \eta, \ket{\psi}\rangle\Downarrow l}
			{\langle \texttt{t}(E), \sigma,\eta, \ket{\psi}\rangle \Downarrow T_{l}\ket{\psi}}	
		\inference{\langle E, \sigma, \eta, \ket{\psi}\rangle\Downarrow l}
			{\langle \texttt{tdg}(E), \sigma,\eta, \ket{\psi}\rangle \Downarrow T_{l}^\dagger\ket{\psi}} \\
		\inference{\langle E_1, \sigma, \eta, \ket{\psi}\rangle\Downarrow l_1 \quad \langle E_2, \sigma, \eta, \ket{\psi}\rangle\Downarrow l_2}
			{\langle \texttt{cx}(E_1, E_2), \sigma, \eta, \ket{\psi}\rangle\Downarrow \cnot_{l_1, l_2}\ket{\psi}} \;
		\inference{\langle E, \sigma, \eta, \ket{\psi}\rangle\Downarrow \lambda x_1,\dots,x_n.U, \\ 
			\langle U\{E_1/x_1, \dots, E_n/x_n\}, \sigma, \eta, \ket{\psi}\rangle \Downarrow \ket{\psi'}}
			{\langle E(E_1,\dots, E_n), \sigma, \eta, \ket{\psi}\rangle \Downarrow \ket{\psi'}} \\
		\inference{\langle U_1, \sigma, \eta, \ket{\psi}\rangle \Downarrow \ket{\psi'} \qquad 
			\langle U_2, \sigma, \eta, \ket{\psi'}\rangle \Downarrow \ket{\psi''}}
			{\langle \texttt{$U_1$; $U_2$}, \sigma, \eta, \ket{\psi}\rangle \Downarrow\ket{\psi''}}
	\end{gather*}
Commands:
	\begin{gather*}
		\inference{l_0,\dots, l_{I-1} \text{ are fresh heap indices}}
			{\langle\texttt{creg $x[I]$}, \sigma, \eta, \ket{\psi}\rangle \Downarrow \langle \sigma[x\leftarrow (l_0,\dots, l_{I-1})],\eta,\ket{\psi}\rangle} \\
		\inference{l_0,\dots, l_{I-1} \text{ are fresh qubit indices}}
			{\langle\texttt{qreg $x[I]$}, \sigma,\eta, \ket{\psi}\rangle \Downarrow \langle \sigma[x\leftarrow (l_0,\dots, l_{I-1})],\eta,\ket{\psi}\rangle} \\
		\inference{}
			{\langle\texttt{gate $x(x_1,\dots, x_n)$ \{ $U$ \}}, \sigma,\eta, \ket{\psi}\rangle \Downarrow 
			\langle \sigma[x\leftarrow \lambda x_1,\dots,x_n.U],\eta,\ket{\psi}\rangle} \\
		\inference{\langle E_1, \sigma, \eta, \ket{\psi}\rangle\Downarrow l_1 \qquad \langle E_2, \sigma, \eta, \ket{\psi}\rangle\Downarrow l_2}				{\langle \texttt{measure $E_1$ -> $E_2$}, \sigma, \eta, \ket{\psi}\rangle \Downarrow 
					\langle\sigma, \eta[l_2\leftarrow 0], P_{l_1}^0\ket{\psi}\rangle} \\
		\inference{\langle E_1, \sigma, \eta, \ket{\psi}\rangle\Downarrow l_1 \qquad \langle E_2, \sigma, \eta, \ket{\psi}\rangle\Downarrow l_2}
			{\langle \texttt{measure $E_1$ -> $E_2$}, \sigma, \eta, \ket{\psi}\rangle \Downarrow 
					\langle\sigma, \eta[l_2\leftarrow 1], P_{l_1}^1\ket{\psi}\rangle} \\
		\inference{\langle E, \sigma, \eta, \ket{\psi}\rangle\Downarrow l}
			{\langle \texttt{reset $E$}, \sigma, \eta, \ket{\psi}\rangle \Downarrow \langle\sigma, \eta, P_{l}^0\ket{\psi}\rangle} \quad
		\inference{\langle E, \sigma, \eta, \ket{\psi}\rangle\Downarrow l \qquad \eta(l)\neq I}
			{\langle \texttt{if($E$==$I$) \{ $U$ \}}, \sigma, \eta, \ket{\psi}\rangle \Downarrow \langle \sigma, \eta, \ket{\psi}\rangle}\\
		\inference{\langle E, \sigma, \eta, \ket{\psi}\rangle\Downarrow l \qquad \eta(l)=I \\ 
			\langle U, \sigma, \eta, \ket{\psi}\rangle \Downarrow \ket{\psi'}}
			{\langle \texttt{if($E$==$I$) \{ $U$ \}}, \sigma, \eta, \ket{\psi}\rangle \Downarrow \langle \sigma, \eta, \ket{\psi'}\rangle} \\
		\inference{\langle C_1, \sigma, \eta, \ket{\psi}\rangle \Downarrow \langle \sigma', \eta', \ket{\psi'}\rangle \qquad 
			\langle C_2, \sigma', \eta', \ket{\psi'}\rangle \Downarrow \langle \sigma'', \eta'', \ket{\psi''}\rangle}
			{\langle \texttt{$C_1$; $C_2$}, \sigma, \eta, \ket{\psi}\rangle \Downarrow \langle \sigma'', \eta'', \ket{\psi''}\rangle}
	\end{gather*}
	\caption{openQASM semantics}\label{fig:semantics}
\end{figure}

\section{Adding types to QASM}\label{sec:tqasm}

Run-time errors may occur in syntactically valid openQASM programs in a number of ways -- particularly when either an array access is out of bounds and the program halts, or a classical (resp. quantum) location is used in a context when a quantum (resp. classical) location is expected. In the official openQASM specification, the latter error is eliminated by the requirement that only (global) variables can be declared as quantum registers may be used as arguments to gates, for instance. In either case however, it is desirable to check that an openQASM program \emph{will not go wrong}, as circuit simulations are frequently run on large, expensive supercomputers (e.g., \cite{hs17}).

In this section we developed a typed variant of openQASM, called typedQASM, which provably rules out such runtime errors. Moreover, the type system uses \emph{sized types} to eliminate out-of-bound accesses, which we later develop into the core of our metaprogramming type system. The use of a type system in this case actually allows \emph{more} valid programs to be written than the standard openQASM specification, as the type system allows us to remove some syntactic distinctions and instead make them in the type system. In particular, our type system allows registers and circuits to be passed as functions to other circuits, whereas the formal specification restricts circuit arguments to only individual qubits.

\begin{figure}[t]
	\begin{tabular}{rrl}
		Base types $\beta$ & $::=$ & \texttt{Bit} $\mid$ \texttt{Qbit} \\
		Types $\tau$ & $::=$ & $\beta$ $\mid$ $\beta[I]$ $\mid$ \texttt{Circuit}$(\tau_1,\dots, \tau_n)$ \\
		Command $C$ & $::=$ & $\dots$ $\mid$ \texttt{creg $x[I]$ in \{ $C$ \}} $\mid$ \texttt{qreg $x[I]$ in \{ $C$ \}} \\
		& $\mid$ & \texttt{gate $x(x_1:\tau_1,\dots, x_n:\tau_n)$ \{ $U$ \} in \{ $C$ \}}
	\end{tabular}
	\caption{typedQASM specification}\label{fig:typedsyntax}
\end{figure}

\Cref{fig:typedsyntax} gives the syntax of typedQASM. We only show the syntactic elements which are different from openQASM or otherwise new. To simplify our analysis, declarations are given explicit block scope, though we leave textual examples in the regular openQASM style of declaration. As the semantics of typedQASM is effectively identical, modulo the block scoping, to openQASM we don't explicitly give the semantics.

\subsection{The type system}

Figure~\ref{fig:typing} gives the rules of our type system. As is standard, the judgement $\Gamma\vdash S:\tau$ states that in the context $\Gamma$ consisting of pairs of identifiers and types, $S$ can be assigned type $\tau$. We overload $\vdash$ to allow environment judgements of the form $\vdash \sigma:\Gamma$ stating that the $\sigma$ maps identifiers $x$ to values of the type $\tau$ if $x:\tau\in\Gamma$.

\begin{figure}[h]
Environment:\vspace*{-1em}
	\begin{gather*}
		\inference{}{\vdash \cdot:\cdot} \quad 
		\inference{\vdash \sigma:\Gamma}{\vdash \sigma[x\leftarrow (l_0,\dots, l_{I-1})]:\Gamma, x:\beta[I]} \\
		\inference{\vdash \sigma:\Gamma \qquad
			\Gamma, x_1:\tau_1,\dots, x_n:\tau_n\vdash U:\texttt{Unit}}
			{\vdash \sigma[x\leftarrow \lambda x_1:\tau_1,\dots, x_n:\tau_n. U]:\Gamma,x:\texttt{Circuit$(\tau_1,\dots,\tau_n)$}}
	\end{gather*}
Expressions:
	\begin{gather*}
		\inference{x:\tau\in\Gamma}{\Gamma\vdash x : \tau} \quad 
		\inference{\Gamma \vdash x:\beta[I'] \qquad I \leq I'-1}{\Gamma\vdash x[I] : \beta} \quad
		\inference{\Gamma \vdash E:\beta[I'] \qquad I \leq I'}{\Gamma\vdash E:\beta[I]}
	\end{gather*}
Unitary statements:
	\begin{gather*}
		\inference{\Gamma\vdash E_1:\texttt{Qbit} \qquad \Gamma\vdash E_2:\texttt{Qbit}}
			{\Gamma\vdash \texttt{cx}(E_1,E_2):\texttt{Unit}} \quad
		\inference{\Gamma\vdash E:\texttt{Qbit} \qquad g\in \{\texttt{h},\texttt{t},\texttt{tdg}\}}{\Gamma\vdash g(E):\texttt{Unit}} \\
		\inference{\Gamma\vdash E:\texttt{Circuit}(\tau_1,\dots, \tau_n) \\
			\Gamma\vdash E_1:\tau_1 \quad \cdots \quad \Gamma\vdash E_n:\tau_n}
			{\Gamma\vdash E(E_1,\dots, E_n):\texttt{Unit}} \quad
		\inference{\Gamma\vdash U_1:\texttt{Unit} \qquad \Gamma\vdash U_2:\texttt{Unit}}{\Gamma\vdash \texttt{$U_1$; $U_2$}:\texttt{Unit}}
	\end{gather*}
Commands:
	\begin{gather*}
		\inference{\Gamma,x:\texttt{Bit}[I]\vdash C:\texttt{Unit}}{\Gamma\vdash \texttt{creg $x[I]$ in \{ $C$ \}}:\texttt{Unit}} \quad
		\inference{\Gamma,x:\texttt{Qbit}[I]\vdash C:\texttt{Unit}}{\Gamma\vdash \texttt{qreg $x[I]$ in \{ $C$ \}}:\texttt{Unit}} \\
		\inference{\Gamma, x_1:\tau_1,\dots,x_n:\tau_n\vdash U:\texttt{Unit} 
			\quad \Gamma,x:\texttt{Circuit($\tau_1,\dots, \tau_n$)}\vdash C:\texttt{Unit}}
			{\Gamma\vdash \texttt{gate $x(x_1:\tau_1,\dots, x_n:\tau_n)$ \{ $U$ \} in \{ $C$ \}}:\texttt{Unit}} \\
		\inference{\Gamma\vdash E_1:\texttt{Qbit} \qquad \Gamma\vdash E_2:\texttt{Bit}}
			{\Gamma\vdash \texttt{measure $E_1$ -> $E_2$}:\texttt{Unit}} \quad
		\inference{\Gamma\vdash E:\texttt{Qbit}}{\Gamma\vdash \texttt{reset $E$}:\texttt{Unit}} \\
		\inference{\Gamma\vdash E:\texttt{Bit} \qquad \Gamma\vdash U:\texttt{Unit}}
			{\Gamma\vdash \texttt{if($E$==$I$) \{ $U$ \}}:\texttt{Unit}} \quad
		\inference{\Gamma\vdash C_1:\texttt{Unit} \qquad \Gamma\vdash C_2:\texttt{Unit}}{\Gamma\vdash \texttt{$C_1$; $C_2$}:\texttt{Unit}}
	\end{gather*}
	\caption{typedQASM typing rules}\label{fig:typing}
\end{figure}

The type system of typedQASM is mostly as expected, with the exception of static-length registers and register bounds checks in the typing rules for dereferences. To give the programmer flexibility to apply gates and circuits to just parts of a larger register -- for instance, when performing an $n$-bit addition into a length $2n$ register as in binary multiplication -- the type system also implicitly supports subtyping of static length registers. Specifically, any length $I$ array can be used in a context requiring \emph{at most} $I$ cells. While this adds a great deal of flexibility on the side of the programmer, as a downside typedQASM typing derivations are not unique. 

As an example of a well-typed QASM program, we show an implementation of the Toffoli circuit from \Cref{fig:toffoli} below:
\begin{lstlisting}
gate toffoli(x:Qbit, y:Qbit, z:Qbit) {
	h(z);
	t(x); t(y); t(z);
	cx(x,y); cx(x,z);
	tdg(y); tdg(z);
	cx(y,z); cx(z,x);
	t(x); tdg(z);
	cx(z,x); cx(x,y); cx(y,z);
	h(z)
}
\end{lstlisting}

\subsection{Type safety}

We now briefly sketch a proof of type safety for typedQASM. In particular, we show that typedQASM is strongly normalizing, as expected.

As is standard, we establish strong normalization by giving type preservation and progress lemmas. While type preservation is effectively implicit in the semantics of typedQASM due to the different syntactic classes, expressions may return different types of values and so we give a form of type preservation for such terms.
\begin{lemma}[Preservation (expressions)]
If $\Gamma\vdash E:\tau$, $\vdash \sigma:\Gamma$ and $\langle E, \sigma, \eta, \ket{\psi}\rangle \Downarrow v$, then either
\begin{itemize}
	\item $\tau=\beta$ and $v=l$ for some base type $\beta$ \& location $l$,
	\item $\tau=\beta[I]$ and $v=(l_0,\dots, l_{I'})$ where $I'\geq I$, or
	\item $\tau=\textnormal{\texttt{Circuit$(\tau_1,\dots,\tau_n)$}}$ and $v=\lambda x_1:\tau_1, \dots, x_n:\tau_n. U$.
\end{itemize}
\end{lemma}
\begin{proof}
If $\tau=\beta$ then we must have $E=x[I]$, hence by the definition of $\Downarrow$, $v=l$. Likewise if $\tau=\beta[I]$ then we must have $E=x$ where $x:\beta[I']\in\Gamma$ for some $I'\geq I$, and since $\vdash \sigma:\Gamma$ then $v=\sigma(x)=(l_0,\dots, l_{I'})$. The case for $\tau=\texttt{Circuit$(\tau_1,\dots,\tau_n)$}$ is similar.
\end{proof}

The following lemmas give progress properties -- the fact that for a well-typed program, evaluation can always continue -- for the different syntactic classes of typedQASM. Together with type preservation, the result is that any well-typed typedQASM program evaluates to a value, i.e. that typedQASM is strongly normalizing.
\begin{lemma}[Progress (expressions)]
If $\Gamma\vdash E:\tau$ and $\vdash \sigma:\Gamma$, then for any $\eta, \ket{\psi}$, $\langle E, \sigma, \eta, \ket{\psi}\rangle\Downarrow v$.
\end{lemma}
\begin{proof}
By case analysis on $E$. If $E=x$ the proof is trivial, as $x:\tau\in\Gamma$ by inversion and $\vdash \sigma:\Gamma$ implies $x\in\textsf{dom}(x)$. If on the other hand $E=x[I]$, we must have $x:\beta[I']\in\Gamma$ for some $I'>I$. Then by preservation, $\langle x, \sigma, \eta, \ket{\psi}\rangle\Downarrow (l_0,\dots, l_{I''})$ for some $I''\geq I'-1$, hence $\langle x, \sigma, \eta, \ket{\psi}\rangle\Downarrow l_I$
\end{proof}

\begin{lemma}[Progress (unitary stmts)]
If $\Gamma\vdash U:\textnormal{\texttt{Unit}}$ and $\vdash \sigma:\Gamma$, then for any $\eta, \ket{\psi}$, $\langle U, \sigma, \eta, \ket{\psi}\rangle\Downarrow \ket{\psi'}$.
\end{lemma}
\begin{proof}
For the case $U=E(E_1,\dots,E_n)$, by the typing derivation we have $\Gamma\vdash E:\texttt{Circuit$(\tau_1,\dots,\tau_n)$}$ so by progress and preservation for expressions, $\langle E, \sigma, \eta, \ket{\psi}\rangle \Downarrow \lambda x_1:\tau_1, \dots, x_n:\tau_n. U$. By the substitution lemma below, $\Gamma\vdash U\{E_1/x_1,\dots, E_n/x_n\}:\textnormal{\texttt{Unit}}$ and hence we can structural induction to show that $\langle U, \sigma, \eta, \ket{\psi}\rangle\Downarrow \ket{\psi'}$.

\begin{lemma}[Substitution]
If $\Gamma, x_1:\tau_1,\dots, x_n:\tau_n\vdash U:\textnormal{\texttt{Unit}}$, and $\Gamma\vdash E_i:\tau_i$ for each $1\leq i\leq n$ then $\Gamma\vdash U\{E_1/x_1,\dots, E_n/x_n\}:\textnormal{\texttt{Unit}}$
\end{lemma}
\end{proof}

\begin{lemma}[Progress (commands)]\label{lem:pcommands}
If $\Gamma\vdash C:\textnormal{\texttt{Unit}}$ and $\vdash \sigma:\Gamma$, then for any $\eta, \ket{\psi}$, $\langle C, \sigma, \eta, \ket{\psi}\rangle\Downarrow \langle \sigma', \eta', \ket{\psi'}\rangle$.
\end{lemma}
\begin{proof}
Proof by induction on the structure of $C$. We show one case: $$C=\texttt{gate $x(x_1:\tau_1,\dots, x_n:\tau_n)$ \{ $U$ \} in \{ $C$ \}}$$

We know that
\[
\resizebox{\linewidth}{!}{$
	\langle \texttt{gate $x(x_1:\tau_1,\dots, x_n:\tau_n)$ \{ $U$ \}}, \sigma, \eta, \ket{\psi}\rangle \Downarrow
	\langle \sigma[x\leftarrow \lambda x_1,\dots,x_n.U],\eta,\ket{\psi}\rangle.
$}
\]
By the typing derivation, $\Gamma,x_1:\tau_1,\dots, x_n:\tau_n\vdash U:\texttt{Unit}$ and $\Gamma,x:\texttt{Circuit$(\tau_1,\dots, \tau_n)$}\vdash C:\texttt{Unit}$. It then follows that \[\vdash \sigma[x\leftarrow \lambda x_1:\tau_1,\dots, x_n:\tau_n. U]:\Gamma,x:\texttt{Circuit$(\tau_1,\dots, \tau_n)$},\] and hence we can apply the inductive hypothesis to complete the case.

The remaining cases are similar.
\end{proof}

\begin{theorem}[Strong normalization]
If $\vdash C:\textnormal{\texttt{Unit}}$, then $$\langle C, \emptyset, \lambda l.0, \ket{00\cdots}\rangle\Downarrow \langle \sigma, \eta, \ket{\psi}\rangle.$$
\end{theorem}
\begin{proof}
Direct consequence of \Cref{lem:pcommands}.
\end{proof}

\section{MetaQASM}\label{sec:mqasm}

Now that we have a safe, array-bounds-checked, typed language, we can add metaprogramming features. In particular, we wish to support\footnote{Controlled circuits are another desirable metaprogramming feature found in many quantum circuit description languages. While metaQASM gates are in fact closed over qubit controls, they require \emph{ancillae} to construct \cite{kmm12}. This complicates the inclusion of a control instruction in metaQASM, and further abstracts away from concrete, resource-driven nature of QASM.}
\begin{itemize}
	\item circuit inversion/reversal, and
	\item circuits parametrized by sizes.
\end{itemize}
While the latter could be accomplished in an ad-hoc way, allowing \emph{type-level} integers allows for more safety in that array bounds can be statically checked, and increases the readability of programs. Moreover, it enforces a clear separation between circuits and families of circuits, which naturally support different operations -- for instance, a family of circuits can't easily be visualized diagrammatically, while a particular instance can \cite{prz17}.

\begin{figure}[t]
	\resizebox{\textwidth}{!}{
	\begin{tabular}{rrl}
		Types $\tau$ & $::=$ & $\dots$ $\mid$ \texttt{Family$(y_1,\dots, y_m)(\tau_1,\dots, \tau_n)$} \\
		Index $I$ & $::=$ & $\dots$ $\mid$ $y$ $\mid$ $\infty$ $\mid$ $I_1 + I_2$ $\mid$ $I_1 - I_2$ $\mid$ $I_1\cdot I_2$ \\
		Range $\iota$ & $::=$ & [$I_1$,$I_2$] \\
		Expression $E$ & $::=$ & $\dots$ $\mid$ \texttt{instance$(I_1, \dots, I_m)$ $E$} \\
		Unitary Stmt $U$ & $::=$ & $\dots$ $\mid$ \texttt{reverse $U$} $\mid$ \texttt{for $y=I_1..I_2$ do \{ $U$ \}} \\
		Command $C$ & $::=$ & $\dots$ $\mid$ 
			\texttt{family$(y_1,\dots, y_m)$ $x(x_1:\tau_1,\dots, x_n:\tau_n)$ \{ $U$ \} in \{ $C$ \}} \\
		Value $V$ & $::=$ & $\dots$ $\mid$ $\Pi y_1,\dots, y_m.V$
	\end{tabular}
	}
	\caption{metaQASM syntax}\label{fig:msyntax}
\end{figure}

Figure~\ref{fig:msyntax} gives the new syntax for metaQASM. Indices $I$ are extended with index variables $y$ and integer arithmetic, and a new syntactic form defining a family of quantum circuits parametrized over index variables is given. The index $\infty$ only exists in the process of type checking and is not valid syntax in source code. Intuitively, the declaration \[\texttt{family$(y_1,\dots, y_m)$ $x(x_1:\tau_1,\dots, x_n:\tau_n)$ \{ $U$ \} in \{ $C$ \}}\] introduces index variables $y_1, \dots, y_m$ into the evaluation and type checking contexts for $\tau_i$ and $U$. 

Figure~\ref{fig:msemantics} gives the semantics of the new syntax. Since index variables cannot be modified or captured, we use a substitution style of evaluation for circuit families. The \texttt{reverse} command introduces a new reduction relation $\langle \texttt{$U$}, \sigma, \ket{\psi}\rangle\Uparrow v$ for which reduction of $U$ is inverted. We give a concrete semantics rather than an abstract rule such as 
\[
	\inference{\langle \texttt{$U$}, \sigma, \eta, \ket{\psi'}\rangle\Downarrow \ket{\psi}}
			{\langle \texttt{reverse $U$}, \sigma, \eta, \ket{\psi}\rangle\Downarrow \ket{\psi'}}
\]
so that metaQASM has a concrete execution model. Inversion of circuits is straightforward in metaQASM, as in any closed context a unitary statement can be statically unrolled to a finite sequence of gates.

\begin{figure}
Indices:
\vspace{-1em}
	\begin{gather*}
		\inference{}{\langle i, \sigma, \eta, \ket{\psi}\rangle\Downarrow i} \qquad
		\inference{\langle I_1, \sigma, \eta, \ket{\psi}\rangle\Downarrow i_1 \quad 
			\langle I_2, \sigma, \eta, \ket{\psi}\rangle\Downarrow i_2 \quad 
			\star\in\{+, -, \cdot\}}
			{\langle I_1\star I_2, \sigma, \eta, \ket{\psi}\rangle\Downarrow i_1 \star i_2}
	\end{gather*}
Expressions:
	\begin{gather*}
		\inference{\langle E, \sigma, \eta, \ket{\psi}\Downarrow \Pi y_1,\dots, y_m.\lambda x_1:\tau_1,\dots, x_n:\tau_n. U}
			{\langle \texttt{instance$(I_1, \dots, I_m)$ $E$}, \sigma, \eta, \ket{\psi} \Downarrow (\lambda x_1:\tau_1,\dots, x_n:\tau_n. U)\{I_1/y_1, \dots, I_m/y_m\}}
	\end{gather*}
Unitary statements:
	\begin{gather*}
		\inference{\langle \texttt{$U$}, \sigma, \eta, \ket{\psi}\rangle\Uparrow \ket{\psi'}}
			{\langle \texttt{reverse $U$}, \sigma, \eta, \ket{\psi}\rangle\Downarrow \ket{\psi'}} \quad
		\inference{\langle I_1, \sigma, \eta, \ket{\psi}\rangle\Downarrow i_1 \quad 
			\langle I_2, \sigma, \eta, \ket{\psi}\rangle\Downarrow i_2 \quad i_1>i_2}
			{\langle \texttt{for $y=I_1..I_2$ do \{ $U$ \}}, \sigma, \eta, \ket{\psi}\rangle \Downarrow \ket{\psi}} \\
		\inference{\langle I_1, \sigma, \eta, \ket{\psi}\rangle\Downarrow i_1 \quad 
			\langle I_2, \sigma, \eta, \ket{\psi}\rangle\Downarrow i_2 \quad i_1\leq i_2 \\
			\langle U\{i_1/y\}, \sigma, \eta, \ket{\psi}\rangle \Downarrow \ket{\psi'} \\
			\langle \texttt{for $y=i_1+1..i_2$ do \{ $U$ \}}, \sigma,\eta, \ket{\psi'}\rangle \Downarrow \ket{\psi''}}
			{\langle \texttt{for $y=I_1..I_2$ do \{ $U$ \}}, \sigma,\eta, \ket{\psi}\rangle \Downarrow \ket{\psi''}}
	\end{gather*}
Reverse reduction:
	\begin{gather*}
		\inference{\langle E, \sigma, \eta, \ket{\psi}\rangle\Downarrow l}
			{\langle \texttt{h$(E)$}, \sigma,\eta, \ket{\psi}\rangle \Uparrow H_{l}\ket{\psi}} \;
		\inference{\langle E, \sigma, \eta, \ket{\psi}\rangle\Downarrow l}
			{\langle \texttt{t$(E)$}, \sigma,\eta, \ket{\psi}\rangle \Uparrow T_{l}^\dagger\ket{\psi}} \;	
		\inference{\langle E, \sigma, \eta, \ket{\psi}\rangle\Downarrow l}
			{\langle \texttt{tdg$(E)$}, \sigma,\eta, \ket{\psi}\rangle \Uparrow T_{l}\ket{\psi}} \\
		\inference{\langle E_1, \sigma, \eta, \ket{\psi}\rangle\Downarrow l_1 \quad \langle E_2, \sigma, \eta, \ket{\psi}\rangle\Downarrow l_2}
			{\langle \texttt{cx$(E_1, E_2)$}, \sigma,\eta, \ket{\psi}\rangle\Uparrow \cnot_{l_1, l_2}\ket{\psi}} \;
		\inference{\langle E, \sigma, \eta, \ket{\psi}\rangle\Downarrow \lambda x_1,\dots,x_n.U, \\ 
			\langle U\{E_1/x_1, \dots, E_n/x_n\}, \sigma, \eta, \ket{\psi}\rangle \Uparrow \ket{\psi'}}
			{\langle E(E_1,\dots, E_n), \sigma,\eta, \ket{\psi}\rangle \Uparrow \ket{\psi'}} \\
		\inference{\langle U_2, \sigma, \eta, \ket{\psi}\rangle \Uparrow \ket{\psi'} \qquad 
			\langle U_1, \sigma, \eta, \ket{\psi'}\rangle \Uparrow \ket{\psi''}}
			{\langle U_1 ;\; U_2, \sigma, \eta, \ket{\psi}\rangle \Uparrow\ket{\psi''}} \\
		\inference{\langle \texttt{$U$}, \sigma, \eta, \ket{\psi}\rangle\Downarrow \ket{\psi'}}
			{\langle \texttt{reverse $U$}, \sigma, \eta, \ket{\psi}\rangle\Uparrow \ket{\psi'}} \quad
		\inference{\langle I_1, \sigma, \eta, \ket{\psi}\rangle\Downarrow i_1 \quad 
			\langle I_2, \sigma, \eta, \ket{\psi}\rangle\Downarrow i_2 \quad i_2<i_1}
			{\langle \texttt{for $y=I_1..I_2$ do \{ $U$ \}}, \sigma, \eta, \ket{\psi}\rangle \Uparrow \ket{\psi}} \\
		\inference{\langle I_1, \sigma, \eta, \ket{\psi}\rangle\Downarrow i_1 \quad 
			\langle I_2, \sigma, \eta, \ket{\psi}\rangle\Downarrow i_2 \quad i_2\geq i_1 \\
			\langle U\{i_2/y\}, \sigma, \eta, \ket{\psi}\rangle \Uparrow \ket{\psi'} \\
			\langle \texttt{for $y=i_1..i_2-1$ do \{ $U$ \}}, \sigma,\eta, \ket{\psi'}\rangle \Uparrow \ket{\psi''}}
			{\langle \texttt{for $y=I_1..I_2$ do \{ $U$ \}}, \sigma,\eta, \ket{\psi}\rangle \Uparrow \ket{\psi''}}
	\end{gather*}
Commands:
	\begin{gather*}
		\inference{\langle C, \sigma[x\leftarrow \Pi y_1,\dots, y_m.\lambda x_1:\tau_1,\dots,x_n:\tau_n.U],\eta,\ket{\psi}\rangle \Downarrow
			\langle \sigma', \eta', \ket{\psi'}\rangle}
			{\langle\texttt{family$(y_1,\dots, y_m)$ $x(x_1:\tau_1,\dots, x_n:\tau_n)$ \{ $U$ \} in \{ $C$ \}}, \sigma,\eta, \ket{\psi}\rangle 
				\Downarrow \langle \sigma, \eta', \ket{\psi'}\rangle}
	\end{gather*}
	\caption{metaQASM semantics}\label{fig:msemantics}
\end{figure}

\begin{figure}
\begin{lstlisting}
include "toffoli.qasm";
gate maj(a:Qbit, b:Qbit, c:Qbit, res:Qbit) {
	toffoli(b, c, res);
	cx(b, c);
	toffoli(a, c, res);
	cx(b, c)
}	
family(n) add(a:Qbit[n], b:Qbit[n], c:Qbit[n], anc:Qbit[n]) {
	cx(a[0], c[0]);
	cx(b[0], c[0]);
	toffoli(a[0], b[0], anc[0]);
	for i=1..n-1 do {
		cx(a[i], c[i]);
		cx(b[i], c[i]);
		cx(anc[i-1], c[i]);
		maj(a[i], b[i], anc[i-1], anc[i])
	}	
}
\end{lstlisting}
\caption{metaQASM implementation of a carry-ripple adder.}\label{fig:adder}
\end{figure}

As an illustration of metaprogramming in metaQASM, \Cref{fig:adder} gives metaQASM code for a simple (non-garbage-cleaning) adder.
Our syntax (and type system) also allows an instance of a family of circuits to accept other circuit families as arguments, a useful feature which allows circuit families to be parametric in the implementation of a sub-routine as shown below (using a minor syntax extension to allow array slicing).

\begin{lstlisting}
family(n) mult(x:Qbit[n], y:Qbit[n], z:Qbit[2*n], 
	anc:Qbit, ctrlAdd:Family(m)
		(x:Qbit, y:Qbit[m], z:Qbit[m], c:Qbit)) 
{
  for i=0..n-1 do {
    instance(n) ctrlAdd(x[i], y, z[i..i+n-1], anc)
  }
}
\end{lstlisting}

By extending our syntax with parametrized gates as in regular openQASM, we can also define a parametrized family of circuits computing the \emph{quantum Fourier transform} as in \cite{prz17}.

\begin{minipage}{\linewidth}
\begin{lstlisting}
include "cphase.qasm";
family(n) qft(x:Qbit[n]) {
	for i=0..n-1 do {
		h(x[i]);
		for j=i+1..n-1 do {
			cphase(j-1+1)(x[i], x[j])
		}
	}
}
\end{lstlisting}
\end{minipage}

\subsection{Type system}

The type system of metaQASM is inspired by Dependent ML \cite{x01}. Figure~\ref{fig:mtype} gives the rules of our system. Type rules are defined over two contexts $\Delta;\Gamma$, where $\Delta$ contains interval constraints on index variables.

\begin{figure}
Indices:
	\begin{gather*}
		\inference{}{\Delta\vdash i:[i,i]} \quad 
		\inference{y:[I_1, I_2]\in\Delta}{\Delta\vdash y:[I_1, I_2]} \quad
		\inference{\Delta\vdash I:[I_1, I_2] \quad \Delta\models I_1'\leq I_1 \quad  \Delta\models I_2'\geq I_2}
			{\Delta\vdash I:[I_1', I_2']} \\
		\inference{\Delta\vdash I:[I_1,I_2] \quad \Delta\vdash I':[I'_1,I'_2]}
			{\Delta\vdash I + I':[I_1+I_1',I_2+I_2']} \quad
		\inference{\Delta\vdash I:[I_1,I_2] \quad \Delta\vdash I':[I_1',I_2']}
			{\Delta\vdash I - I':[I_1-I_1',I_2-I_2']} \quad \\
		\inference{\Delta\vdash I:[I_1,I_2] \quad \Delta\vdash I':[I_1',I_2'] \\
			\Delta\models I_1''=\min(I_1\cdot I_1', I_1\cdot I_2', I_2\cdot I_1', I_2\cdot I_2') \\ 
			\Delta\models I_2''=\max(I_1\cdot I_1', I_1\cdot I_2', I_2\cdot I_1', I_2\cdot I_2')}
			{\Delta\vdash I \cdot I':[I_1'', I_2'']} \;
	\end{gather*}
Expressions:
	\begin{gather*}
		\inference{\Delta;\Gamma \vdash x:\beta[I'] \qquad \Delta\models 0\leq I< I'}
			{\Delta;\Gamma\vdash x[I] : \beta} \\
		\inference{\Delta;\Gamma \vdash E:\texttt{Family$(y_1,\dots, y_m)(\tau_1,\dots, \tau_n)$} \\
			\Delta\vdash I_1:[0,\infty] \quad \cdots \quad \Delta\vdash I_m:[0,\infty]}
			{\resizebox{\textwidth}{!}{$\Delta;\Gamma\vdash \texttt{instance$(I_1, \dots, I_m)$ $E$}:\texttt{Circuit($\tau_1\{I_1/y_1,\dots, I_m/y_m\},\dots, \tau_n\{I_1/y_1,\dots, I_m/y_m\}$)}$}}
	\end{gather*}
Unitary statements:
	\begin{gather*}
		\inference{\Delta;\Gamma\vdash U:\texttt{Unit}}{\Delta;\Gamma\vdash \texttt{reverse $U$}:\texttt{Unit}} \quad
		\inference{\Delta\vdash I:[I_1, I_2] \qquad \Delta\vdash I':[I_1',I_2'] \\
			\Delta, y:[I,I']; \Gamma\vdash U:\texttt{Unit}}
			{\Delta;\Gamma\vdash \texttt{for $y=I..I'$ do \{ $U$ \}}:\texttt{Unit}}
	\end{gather*}
Commands:
	\begin{gather*}
		\inference{\Delta, y_1:[0,\infty],\dots, y_m:[0,\infty]\vdash \tau_1 :: *\quad \cdots \quad 
			\Delta, y_1:[0,\infty],\dots, y_m:[0,\infty]\vdash \tau_n :: * \\
			\Delta, y_1:[0,\infty],\dots, y_m:[0,\infty]; \Gamma, x_1:\tau_1,\dots,x_n:\tau_n\vdash U:\texttt{Unit}, \\
			\Delta;\Gamma,x:\texttt{Family$(y_1,\dots, y_m)(\tau_1,\dots, \tau_n)$}\vdash C:\texttt{Unit}}
			{\Delta;\Gamma\vdash \texttt{family$(y_1,\dots, y_m)$ $x(x_1:\tau_1,\dots, x_n:\tau_n)$ \{ $U$ \} in \{ $C$ \}}:\texttt{Unit}} \\
	\end{gather*}
	\caption{metaQASM typing rules}\label{fig:mtype}
\end{figure}

As with typedQASM, array bounds are checked and subtyping on array lengths is allowed. Integer expressions are assigned intervals which may be arbitrary (well-formed) integer expressions. The judgement $\Delta\models P$ which appears in the typing rules for integer expressions denotes that under the context $\Delta$, the (in)equality $P$ holds. We leave a particular constraint solver up to implementation. It remains an open question whether undecidable constraints can be generated by our type system, though in practice it appears most common constraints can be efficiently solved with off-the-shelf constraint solvers \cite{x01}.

The type system of \Cref{fig:mtype} also involves \emph{kind} judgements of the form
\[
	\Delta\vdash \tau :: *
\]
stating that $\tau$ is a simple type in the index context $\Delta$. While the rules of our kind system are not given here, it is straightforward to derive. In particular, $\tau$ has kind $*$ if $\tau$ does not reference any free index variables, and does not contain any registers of negative length.

\begin{remark}
The fact that metaQASM has no means of specifying and checking relational properties on indices causes some programs to require counter-intuitive type schemes. For instance, the following $n$-bit adder is not well-typed due to the statement \texttt{toffoli(x[n-2], ctrl, y[n-1])}, though it does not cause run-time errors when $n\geq 2$.
\begin{lstlisting}
include "toffoli.qasm";
family(n) ctrlAdd(ctrl:Qbit, x:Qbit[n], 
                  y:Qbit[n], c:Qbit) {
	toffoli(x[0], ctrl, y[0]);
	cx(x[0], c);
	toffoli(c, y[0], x[0]);
	for i=1..n-2 do {
		toffoli(x[i], ctrl, y[i]);
		cx(x[i-1], x[i]);
		toffoli(x[i-1], y[i], x[i])
	}
	toffoli(x[n-1], ctrl, y[n-1]);
	toffoli(x[n-2], ctrl, y[n-1]);
	for i=2..n-1 do {
		toffoli(x[n-i-1], y[n-i], x[n-i]);
		cx(x[n-i-1], x[n-i]);
		toffoli(x[n-i-1], ctrl, y[n-i])
	}
	toffoli(c, y[0], x[0]);
	cx(x[0], c);
	toffoli(c, ctrl, y[0])
}
\end{lstlisting}
The above adder can modified \cite{f18} to a well-typed program by using $m=n-2$ as the parameter, effectively specifying the number of entries \emph{greater than $2$} that the input registers contain. The program snippet below gives the declaration required to make the controlled Adder implementation (with appropriate re-indexing) well-typed.
\begin{lstlisting}
family(m) ctrlAdd(ctrl:Qbit, x:Qbit[m+2], 
                  y:Qbit[m+2], c:Qbit)
\end{lstlisting}
\end{remark}

In most practical cases appropriate parameters can be given so as to allow a well-typed implementation of a circuit family. However, the family parameters can be counter-intuitive, and more egregiously it can be unclear as to how to generate an intended instance. We leave it as an avenue for future work to add specification and checking of bounds and relational properties to metaQASM.

\subsection{Type safety}

As in the case of typedQASM, metaQASM is strongly normalizing, due to the lack of recursion and unbounded loops. Progress relies on the fact that during the course of evaluation, no free index variables are encountered -- hence any term encountered by an interpreter is well-typed in the empty index context, and in particular indices can be evaluated to finite integers, as shown below.
\begin{lemma}\label{lem:pn}
If $\cdot \vdash I:[I_1,I_2]$, then $\langle I, \sigma, \eta, \ket{\psi}\rangle\Downarrow i$.
\end{lemma}
\begin{proof}
Trivial since the judgement $\cdot \vdash I:[I_1,I_2]$ requires that $I$ does not contain any variables. Note also that there is no derivation of a judgement of the form $\Delta\vdash \infty:[I_1,I_2]$ hence $I$ cannot contain any infinite integers.
\end{proof}

The remaining lemmas are extensions of results for typedQASM. Only the new or different cases are considered.
\begin{lemma}[Preservation (expressions)]
If $\cdot;\Gamma\vdash E:\tau$, $\vdash \sigma:\Gamma$ and $\langle E, \sigma, \eta, \ket{\psi}\rangle \Downarrow v$, then either
\begin{enumerate}
	\item $\tau=\beta$ and $v=l$,
	\item $\tau=\beta[I]$ and $v=(l_0,\dots, l_{I'})$ where $I'\geq I$, or
	\item $\tau=\textnormal{\texttt{Circuit$(\tau_1,\dots,\tau_n)$}}$ and $v=\lambda x_1:\tau_1, \dots, x_n:\tau_n. U$
	\item $\tau=\textnormal{\texttt{Family$(y_1,\dots, y_m)(\tau_1,\dots, \tau_n)$}}$ and 
	\[v=\Pi y_1,\dots, y_m.\lambda x_1:\tau_1, \dots, x_n:\tau_n. U.\]
\end{enumerate}
\end{lemma}
\begin{proof}
The new \texttt{Family} case is effectively identical to the \texttt{Circuit} case. For the case where $\tau=\beta$, it suffices to note that by Lemma~\ref{lem:pn}, the expressions $I$ and $I'$ in the typing derivation reduce to integers $i, i'$ and the proof concludes as in the typedQASM case. 

Finally we have to revise the $\tau=\textnormal{\texttt{Circuit$(\tau_1,\dots,\tau_n)$}}$ case as we now have two possible derivations. The new case $E=\texttt{instance$(I_1, \dots, I_m)$ $E$}$ is also trivial as the only reduction produces a value of the form $\lambda x_1:\tau_1, \dots, x_n:\tau_n. U$. Note that the type $\tau$ in the derivation has $I_i$ substituted for index variables $y_i$, as in the conclusion of the reduction rule.
\end{proof}

\begin{lemma}[Progress (expressions)]
If $\cdot;\Gamma\vdash E:\tau$ and $\vdash \sigma:\Gamma$, then for any $\eta, \ket{\psi}$, $\langle E, \sigma, \eta, \ket{\psi}\rangle\Downarrow v$.
\end{lemma}
\begin{proof}
Again, the new case $E=\texttt{instance$(I_1, \dots, I_m)$ $E$}$ needs consideration. By inversion we see that $E:\texttt{Family$(y_1,\dots, y_m)(\tau_1,\dots, \tau_n)$}$. By structural induction and the preservation lemma, $\langle E', \sigma, \eta, \ket{\psi}\rangle\Downarrow \Pi y_1,\dots, y_m.\lambda x_1:\tau_1, \dots, x_n:\tau_n. U$ and so $\langle E, \sigma, \eta, \ket{\psi}\rangle\Downarrow v$.
\end{proof}

\begin{lemma}[Progress (unitary stmts)]
If $\cdot;\Gamma\vdash U:\textnormal{\texttt{Unit}}$ and $\vdash \sigma:\Gamma$, then for any $\eta, \ket{\psi}$, $\langle U, \sigma, \eta, \ket{\psi}\rangle\Downarrow \ket{\psi'}$.
\end{lemma}
\begin{proof}
The case $U=\texttt{reverse $U$}$ requires a separate progress lemma for reverse reduction, which follows similar to progress for unitary statements.

For the remaining case $U=\texttt{for $y=I_1..I_2$ do \{ $U$ \}}$, it suffices to observe that by inversion, $\cdot \vdash I_i:[I_i,I_i']$ and so both bounds reduce to integers. As each recursive call increases the lower bound $I_1$, and $I_2$ is necessarily finite, there can be no infinite chains of reductions. The only condition that needs checking is that $\langle U\{i_1/y\}, \sigma, \eta, \ket{\psi}\rangle \Downarrow \ket{\psi'}$, for which we need the following substitution lemma.

\begin{lemma}
If $\Delta, y:[I_1, I_2'];\Gamma\vdash U:\textnormal{\texttt{Unit}}$, and $\Delta\vdash i_1:[I_1,I_2']$ then $\Delta;\Gamma\vdash U\{i_1/y\}:\textnormal{\texttt{Unit}}$
\end{lemma}
To complete the proof, another lemma is needed stating that the result of evaluating an integer expression is within the bounds of the expression's type. We leave this as an easy exercise.
\end{proof}

\begin{lemma}[Progress (commands)]\label{lem:pcommandsfull}
If $\cdot;\Gamma\vdash C:\textnormal{\texttt{Unit}}$ and $\cdot;\cdot \vdash \sigma:\Gamma$, then for any $\eta, \ket{\psi}$, $\langle C, \sigma, \eta, \ket{\psi}\rangle\Downarrow \langle \sigma', \eta', \ket{\psi'}\rangle$.
\end{lemma}
\begin{proof}
We have one new command to check,
\[
	C=\texttt{family$(y_1,\dots, y_m)$ $x(x_1:\tau_1,\dots, x_n:\tau_n)$ \{ $U$ \} in \{ $C$ \}}.
\] The proof in this case is effectively identical to regular gate declaration.
\end{proof}

\begin{theorem}[Strong normalization]
If $\cdot;\cdot \vdash C:\textnormal{\texttt{Unit}}$, then \[\langle C, \emptyset, \lambda l.0, \ket{00\cdots}\rangle\Downarrow \langle \sigma, \eta, \ket{\psi}\rangle.\]
\end{theorem}
\begin{proof}
	Follows directly from \Cref{lem:pcommandsfull}
\end{proof}

\section{Conclusion}\label{sec:conclusion}

We have described a typed extension to openQASM that supports static array bounds checking, higher-order circuits, and lightweight metaprogramming in the form of size-indexed families of circuits. The resulting language is powerful enough to use for writing libraries of general quantum circuit families, such as for reversible arithmetic, while low-level enough to be used wherever openQASM is used.

As this is preliminary work, much remains to be done to make metaQASM a practical language for quantum library development. In particular, a concrete implementation needs to be developed, as do more examples of practical circuit families. A major question which remains is whether a decision procedure for the simple, non-linear integer constraints generated by our type system exists. 

Another interesting question for future work is whether \emph{parametrized resource counts} for algorithms can be computed directly from metaQASM programs. In particular, a desirable feature would be to compute closed-form formulas for the number of qubits, gates, etc., in an arbitrary instance of a circuit family, so that different implementations of the same circuit family can be analytically compared \emph{for any instance size}. Doing so would help not only with resource estimation, but also compilation by allowing compilers to automatically select the best implementation for a particular cost model.

\section*{Acknowledgements}
The author wishes to thank Gregor Richards for motivating this project and Frank Fu for pointing out alternative ways of typing several examples in this manuscript. The author also wishes to thank the anonymous reviewers for their detailed comments which have vastly improved the presentation of this work.

\bibliographystyle{splncs04}
\bibliography{main}

\end{document}